\documentclass[conference]{IEEEtran}
\input epsf
\usepackage{graphicx,amsmath, amssymb, amsthm, mathtools,fancyhdr,float,graphicx,subcaption,enumitem}
\usepackage{commath,balance}
\newcommand{\df}{\mathrm{d}}

\theoremstyle{definition} 
\newtheorem{theorem}{Theorem}[section]

\newtheorem{lemma}[theorem]{Lemma}
\newtheorem{corollary}[theorem]{Corollary}
\newtheorem{proposition}[theorem]{Proposition}

\newtheorem{remark}[theorem]{Remark}
\renewcommand{\IEEEkeywordsname}{Keywords}

\begin{document}

\title{\LARGE \textbf{Autocovariance and Optimal Design for Random Walk Metropolis--Hastings Algorithm}}
\IEEEoverridecommandlockouts

\author{
  Jingyi Zhang and
  James C. Spall\IEEEauthorrefmark{1}\thanks{Corresponding author: James C. Spall (email: James.Spall@jhuapl.edu)}
}

\maketitle

\begin{abstract}
The Metropolis–Hastings algorithm has been extensively studied in the estimation and simulation literature, with most prior work focusing on convergence rates and asymptotic theory. In contrast, the covariance structure of the resulting Markov chain—an important determinant of statistical efficiency—has received comparatively little theoretical attention. In this work, we study the covariance properties of Metropolis–Hastings algorithms in both scalar and high-dimensional settings. For the scalar case, we focus on symmetric unimodal target distributions with symmetric random-walk proposals, and prove positivity of the unit-lag covariance, together with explicit structural characterizations. Within this class, we further establish an optimal proposal design that minimizes covariance and improves sampling efficiency. In the high-dimensional regime, we relate the covariance matrix to the classical 0.23 average acceptance-rate tuning criterion, providing a new covariance-based perspective on optimal scaling.
\end{abstract}

\IEEEoverridecommandlockouts
\IEEEpeerreviewmaketitle

\section{Introduction}
Markov chain Monte Carlo (MCMC) provides a general framework for sampling from complex distributions by constructing a Markov chain with the desired stationary distribution \cite{Spall2003}. This approach enables efficient approximation of expectations and integrals that are otherwise analytically intractable, and it has found widespread applications in statistics, engineering, and control \cite{HillandSpall}. 
In particular, Monte Carlo particle filters can discretize the state space to enable efficient MAP sequence estimation via dynamic programming \cite{Godsill2001}. MCMC has also been applied to Bayesian inference for stochastic kinetic models by sampling joint posteriors of parameters and latent data \cite{Golightly2005}, and to system identification for Bayesian parameter estimation in complex dynamical models beyond the reach of traditional techniques \cite{Ninness2010}. 
In the control domain, particle MCMC has been applied to infer latent trajectories and nonparametric dynamics in Gaussian process state-space models \cite{Frigola2013}, MCMC has been used to generate model uncertainty sets in nonlinear dynamical systems via Hamiltonian Monte Carlo sampling of high-dimensional transfer operators \cite{Srinivasan}, and ghost sampling, an MCMC-based technique, has been introduced to efficiently simulate power disturbances conditional on rate-of-change-of-frequency violations in power systems \cite{Moriarty}; moreover, MCMC has also been integrated with Bernstein approximation techniques to handle non-affine and dependent chance constrained optimization problems, providing an efficient approach to high-dimensional integration \cite{Zhao}.

The Metropolis-–Hastings algorithm is one of the foundational methods in MCMC, with widespread applications in model estimation, physics, and system identification \cite{HillandSpall,Metropolis1953}. The theoretical study of MH has been extensive, with a large body of work devoted to convergence behavior and asymptotic properties of estimators \cite{Tierney1994, RobertsandRosenthal, BrooksandGareth1998}. 
An important line of research in MCMC has focused on the 
scaling and efficiency of Metropolis-Hastings algorithms in high-dimensional settings. For Gaussian random 
walk Metropolis--Hastings (RWMH) algorithms, the efficiency of the chain is maximized when the 
average acceptance rate is tuned to approximately $0.23$ \cite{RobertsandGelman}. This ``$0.23$ rule'' has since 
become a widely used guideline for practitioners, providing a concrete and interpretable 
criterion for tuning proposal variances in high-dimensional problems.

While convergence and asymptotics have been well developed, comparatively little is known about the covariance structure, or equivalently the autocorrelation, of the MH chain itself. 
MCMC methods are closely connected with statistical estimation, providing a means to approximate otherwise intractable quantities. 
Within this framework, the covariance structure of the samples plays a central role in determining estimation accuracy and efficiency. Further knowledge of the covariances of the samples relates strategies for variance reduction, and potentially guides practical implementation of MCMC algorithms.

This work aims to address this gap. 
We first study the scalar case of MH with symmetric random walk proposals, focusing on symmetric unimodal target distributions. In this setting, we establish new theoretical results on the covariance structure and propose an optimal design for the proposal distribution of symmetric random walk form. In addition, we present some general results beyond this restrictive setting. For the high-dimensional case, 
under the same setup as Roberts, Gelman and Gilks \cite{RobertsandGelman}, 
we relate the unit-lag covariance matrix to the classical 0.23 average acceptance rate criterion, providing new theoretical connections between covariance properties and established tuning principles.

Before beginning the technical analysis, we emphasize that Metropolis-Hastings chains are not necessarily positively correlated. 
As a counterexample, consider the target density $\pi(x)=\mathcal{N}(0,1)$ and the proposal density 
$q(y|x)=\mathcal{N}(-c x,2)$ with $c\in(0,1)$. 
In this setting, proposed moves tend to push the chain in the opposite direction of the current state, 
so that $x_t>0$ typically leads to $x_{t+1}<0$ and vice versa. 
Since the covariance reduces to $\mathrm{Cov}(X_t,X_{t+1})=\mathbb{E}[X_t X_{t+1}]$ under the symmetry of $\pi$, 
the systematic flipping effect yields negative correlation. Both theoretical analysis and numerical evaluations confirm that the covariance is strictly negative for all $c\in(0,1)$.

The paper is organized as follows. Section~II presents the preliminary background, notation, and problem setup. Section~III contains the main theoretical results. Section~IV discusses the optimal proposal design. Section~V provides numerical experiments to assess whether the proposed design can improve the efficiency of the Metropolis-Hastings algorithm. Section~VI concludes with remarks and directions for future work.

\section{Preliminaries}
Monte Carlo methods provide a general tool for estimation and simulation when analytical solutions are unavailable. 
MCMC extends this idea by constructing a Markov chain with the desired stationary distribution, enabling sampling from complex or high-dimensional models \cite{Spall2003,HillandSpall}. 
Within this framework, the Metropolis-Hastings (MH) algorithm stands out as a versatile and widely used scheme. In this section, we present the Metropolis-Hastings algorithm, establish the notation, and outline the problem setup.

\subsection{Metropolis-Hastings algorithm}
Assume a continuous state space with probability density functions. The overall aim of MCMC methods such as Metropolis-Hastings algorithm is to generate samples from a target density of interest, say \(\pi(x)\). The MH algorithm is defined by a proposal density $q(y|x)$ for transitioning from state $x$ to state $y$, and a target density $\pi$ of interest.
The MH algorithm accepts a proposed move with probability
\[
\alpha(x,y) = \min\!\left(1, \frac{\pi(y)q(x|y)}{\pi(x)q(y|x)}\right),
\]
and otherwise the chain remains at $x$. 
This induces a Markov transition kernel
\[
\mathcal{K}(x,y) = \alpha(x,y) q(y|x) + \delta(y-x)\big(1-r(x)\big),
\]
where $\delta(\cdot)$ is the Dirac measure at $0$, and
\(
r(x) = \int \alpha(x,y)\,q(y|x)\,\df y
\)
is the overall acceptance probability at current state \cite{HillandSpall}. Under ergodicity and irreducibility assumptions, 
the Markov chain converges to the unique invariant distribution. 
Let $(\mathcal{X},\mathcal{B})$ denote the underlying measurable state space, 
with $\mathcal{X}$ the state space and $\mathcal{B}$ its associated $\sigma$-algebra, then
\[
\pi(y)=\int_{\mathcal{X}}\mathcal{K}(x,y)\,\pi(dx)
\quad \pi\text{-a.s. on }(\mathcal{X},\mathcal{B}).
\]
where ``$\pi$-a.s.'' indicates that the equality holds for all $y$ except on a $\pi$-null set (i.e. almost surely).

The above formulation specifies the acceptance mechanism and transition kernel in mathematical terms. 
For clarity, we now summarize the MH algorithm in a step-by-step procedure.

\begin{enumerate}[leftmargin=3.7em,label=\textbf{Step \arabic*.}]
  \item \textbf{Initialization.} Choose the initial state $X_0 \in \mathcal{X}$ and set $t=0$.
  \item Given $X_t$, generate a candidate $Y$ according to the proposal density $q(\cdot | X_t)$.
  \item Compute the acceptance probability
  \[
    \alpha(X_t,Y) = \min\!\left(1, \frac{\pi(Y)q(X_t | Y)}{\pi(X_t)q(Y | X_t)}\right).
  \]
  Then set
  \[
    X_{t+1} =
      \begin{cases}
        Y, & \text{with probability } \alpha(X_t,Y), \\[6pt]
        X_t, & \text{with probability } 1-\alpha(X_t,Y).
      \end{cases}
  \]
  Increment $t \gets t+1$.
  \item Repeat Steps 2-3 until the terminal iteration $n$ is reached, yielding the chain $(X_0,\dots,X_n)$.
\end{enumerate}

Our aim in this paper is to analyze the autocovariance structure of the sequence \(\{X_0,X_1,X_2,\ldots,X_t\}\) produced by the MH algorithm. We focus on the larget \(t\) (post-``burn-in'') where the process has achieved a stationary probability distribution.

\subsection{Notation}
Let $X_t$ denote the current state. We consider the continuous state spaces, 
where $X_{t+k} \in \mathbb{R}$. 
The target distribution $\pi(\cdot)$ is assumed to have mean $\mu$ and strictly positive finite variance $\sigma_{\pi}^2 \in (0,\infty)$.
 The proposal distribution follows a random walk of the form \(X_{t+1} = X_t + Z, \, Z \sim \phi(\cdot)\), where $\phi(\cdot)$ denotes a density symmetric about 0, not restricted to the Gaussian case.

Throughout this paper we assume that the target distribution $\pi$ admits a density 
with respect to Lebesgue measure. 
Similarly, the proposal distribution $q(\cdot|x)$ is assumed to admit a density $q(y|x)$ 
with respect to the corresponding base measure. $\Pi(\cdot)$ is the cumulative distribution function corresponding to $\pi(\cdot)$, and we will use the cumulative distribution function (CDF) notation for brevity in what follows. Unless otherwise stated, all integrals are taken over $\mathbb{R}$ with respect to Lebesgue measure. 
We also adopt the standard stochastic order notation: 
$Z_n=o_p(a_n)$ means $Z_n/a_n \to 0$ in probability, 
while $Z_n=O_p(a_n)$ means $Z_n/a_n$ is bounded in probability. Also for simplicity, let $\partial_x f(x)$ and $\partial_x^2 f(x)$ be the first and second derivatives of a function $f(x)$ with respect to its argument. Lastly, for a matrix $A$, $\|A\|_{\mathrm{op}}
:= \max\{\,\|Ax\|_2 : x\in\mathbb{R}^n,\ \|x\|_2=1\,\}$ denotes its operator (spectral) norm.

\subsection{Problem Setup}
Consider the unit-lag covariance of the output sequence of the MH algorithm under the stationarity assumption, 
i.e., the current state $X_t \sim \pi(\cdot)$, which implies that the next state $X_{t+1} \sim \pi(\cdot)$ as well. 
Stationarity holds provided that the chain is irreducible, aperiodic, and satisfies detailed balance with respect to $\pi(\cdot)$, 
ensuring that $\pi(\cdot)$ is the unique invariant distribution \cite{HillandSpall}; these conditions are known to hold for an MH algorithm.

Our goal is to investigate the multi-lag covariance of a symmetric random walk MH chain, 
with particular interest in the structure of the covariance and the optimal design of the proposal distribution. 
We begin with the unit-lag covariance, $\mathrm{Cov}(X_t, X_{t+1})$.

\section{Main results}
In this section, we present our main theoretical results, deriving an analytic expression for the unit-lag covariance of a Metropolis-Hastings chain with a symmetric random-walk proposal and a symmetric unimodal target. This serves as the foundation for the optimal proposal design studied in Section~IV.
\begin{lemma}\label{covariance formula}
    
For a Metropolis-Hastings chain, the general formula for the unit-lag covariance is given by
\begin{align}\label{general formula}
        &\mathrm{Cov}(X_t,X_{t+1})\nonumber\\
        &=\sigma_{\pi}^2-\tfrac12\iint(x-y)^2\pi(x)q(y|x)\alpha(x,y)\,\df x\,\df y.
\end{align}
Recall that, unless otherwise specified, all integrals are taken over $\mathbb{R}$.
\end{lemma}
\begin{proof} Using tower property,
\begin{align*}
    \mathrm{Cov}(X_t,X_{t+1})&=\mathbb{E}[X_tX_{t+1}]-\mathbb{E}[X_t]\mathbb{E}[X_{t+1}]\\
    &=\mathbb{E}[\mathbb{E}[X_{t+1}|X_t]\cdot X_{t}]-\mathbb{E}[X_t]\mathbb{E}[X_{t+1}].  
\end{align*}
By stationariy, we know \(\mathbb{E}[X_t]=\mathbb{E}[X_{t+1}]=\mu\). Then computing \(\mathbb{E}[X_{t+1}|X_t=x]\),
\begin{align*}
\mathbb{E}[X_{t+1}|X_t=x]&=\int y\cdot\mathcal{K}(x,y)\,\df y\\
&=\int y\cdot\alpha(x,y)q(y|x)\,\df y+x\cdot(1-r(x))\\
&=x+\int(y-x)\alpha(x,y)q(y|x)\,\df y.
\end{align*}
Then by reversibility,\,i.e.\,\(\pi(x)q(y|x)\alpha(x,y)=\pi(y)q(x|y)\alpha(y,x)\),
\begin{align}\label{first term of cov}
    &\mathbb{E}[\mathbb{E}[X_{t+1}|X_t]\cdot X_{t}]\nonumber\\
    &=\iint x(y-x)\alpha(x,y)q(y|x)\pi(x)\,\df y\,\df x+\int x^2\pi(x)\,\df x\nonumber\\
    &=-\frac12\iint(x-y)^2\alpha(x,y)q(y|x)\pi(x)\,\df y\,\df x+\int x^2\pi(x)\,\df x.
\end{align}
Applying the Tonelli's theorem \cite{Tonelli} to justify the interchange of integrals in \eqref{first term of cov} since the integrand \((x-y)^2\alpha(x,y)q(y|x)\pi(x)\) is nonnegative, and using stationarity, we obtain \eqref{general formula}.
\end{proof}

If we further suppose the proposal density is of random walk form (not necessarily symmetric), \eqref{general formula} simplifies to
\begin{align}\label{formula for rw}
        &\mathrm{Cov}(X_t,X_{t+1})\nonumber\\
        &=\sigma_{\pi}^2-\frac12\iint(x-y)^2\phi(y-x)\pi(x)\alpha(x,y)\,\df x\,\df y
\end{align}
since the proposal density \(q(\cdot|x)\) is of random-walk form, i.e., \(Y=X+Z\) with \(Z\sim\phi(\cdot)\), we have 
\(q(y|x)=\phi(y-x)\) and \(q(x|y)=\phi(x-y)\). Substituting these expressions into \eqref{general formula} yields the desired result.

Moreover, when the proposal density is of symmetric random-walk form (not necessarily Gaussian), 
that is, $q(y|x)=q(x|y)=\phi(|x-y|)$, equation~\eqref{general formula} takes the following form
\begin{align}\label{formula for srw}
&\mathrm{Cov}(X_t,X_{t+1})\nonumber\\
&= \sigma_{\pi}^2
-\tfrac{1}{2}\iint (x-y)^2 \phi(y-x)\min\{\pi(x),\pi(y)\}\,\df x\,\df y.    
\end{align}

Lemma~\ref{covariance formula} expresses the unit-lag covariance in the quadratic form $(x-y)^2$, which naturally connects to random-walk proposals and enables explicit analysis. Building on this, Corollary~\ref{invariance} follows when the proposal is further assumed to be of random-walk form.

\begin{corollary}\label{invariance}Suppose the proposal density is random-walk (not necessarily symmetric), then the unit-lag covariance is invariant under translation of the target density. Define $\pi_c(x):=\pi(x-c),\,c\in\mathbb{R}$, then
\[
\mathrm{Cov}_{\pi_c}(X_t,X_{t+1})=\mathrm{Cov}_{\pi}(X_t,X_{t+1}).
\]
\begin{proof}
Note that the variance is invariant under translation, i.e.,
$\sigma^2_{\pi_c} = \sigma^2_\pi$.
And recalling that the proposal $q(\cdot|x)$ is of the random walk form, $q(y|x)=\phi(y-x)$ and \(q(x|y)=\phi(x-y)\). Thus by definition, the acceptance rate under the shifted target density \(\pi_c(\cdot)\) is
\begin{align*}
    \alpha_c(x,y)
    &= \min\bigg\{1,\frac{\pi(y-c)\,\phi(x-y)}{\pi(x-c)\,\phi(y-x)}\bigg\}.
\end{align*}
Recall from \eqref{general formula} that 
$\mathrm{Cov}_{\pi_c}(X_t,X_{t+1})$ is 
\begin{align*}
    &\mathrm{Cov}_{\pi_c}(X_t,X_{t+1})\\
    &= \sigma_{\pi}^2 - \tfrac{1}{2}\iint (x-y)^2 \phi(x-y)\pi_c(x)\alpha_c(x,y)\,\df x\,\df y \\
    &=\sigma_{\pi}^2 - \tfrac{1}{2}\iint (x-y)^2 \phi(x-y)\pi(x-c)\alpha(x-c,y-c)\,\df x\,\df y\\
    &=  \sigma_{\pi}^2 - \tfrac{1}{2}\iint (s-t)^2 \phi(s-t)\pi(s)\alpha(s,t)\,\df s\,\df t\\&=\mathrm{Cov}_{\pi}(X_t,X_{t+1})
\end{align*}
where we applied the change of variables $x-c=s$, $y-c=t$, noting that the variance is invariant under translation. 
Hence we may simply denote $\mathrm{Cov}_{\pi_c}(X_t,X_{t+1})$ by $\mathrm{Cov}_{\pi}(X_t,X_{t+1})$.
\end{proof}    
\end{corollary}

The argument in Corollary~\ref{invariance} extends naturally to higher lags. That is, for lag $k$, the $k$-step transition kernel remains invariant 
under translations of the target density $\pi$, 
so the covariance is also invariant under translations of the target density. Corollary~\ref{invariance} also allows us, without loss of generality, to assume a zero mean ($\mu=0$) when considering symmetric unimodal target densities later.

\begin{lemma}\label{leq 0.5}
Let \( \phi(x) \) be any probability density on $\mathbb{R}$ that is symmetric and unimodal about 0.
Then,
\[
    |x|\phi(x) < \frac{1}{2} \quad \text{a.e.}
\]
\end{lemma}

\begin{proof}
Suppose that the inequality \( |x|\phi(x) < 1/2 \) fails on a set of positive measure.
This implies that the set is non-empty, so there exists at least one point \( x_0 \neq 0 \) such that
\[
|x_0|\phi(x_0) \geq \frac{1}{2}.
\]
By symmetry and unimodality, \( \phi(|x|) \) is non-increasing for \( x \in \mathbb{R} \). Consequently, for any \( x \) in the interval \( (-|x_0|, |x_0|) \), we must have \( \phi(x) \geq \phi(x_0) \). Integrating \( \phi \) over this interval yields
\[
\int_{-|x_0|}^{|x_0|} \phi(x)\,dx \geq \int_{-|x_0|}^{|x_0|} \phi(x_0)\,dx = 2|x_0|\phi(x_0) \geq 1.
\]
Since \( \phi \) is a probability density function, we know that \( \int_{\mathbb{R}} \phi(x)\,dx = 1 \). The inequality above implies that all the probability mass must be concentrated within \( [-|x_0|, |x_0|] \), and furthermore, that \( \phi(x) \) must be constant on this interval to avoid the integral strictly exceeding 1.

Thus, \( \phi \) must be the uniform distribution on \( [-|x_0|, |x_0|] \). The condition \( |x|\phi(x) \geq 1/2 \) holds only at the boundaries \( x = \pm x_0 \), which constitute a set of measure zero. This contradicts the original assumption.
\end{proof}

Lemma~\ref{leq 0.5} is powerful in that it holds for any symmetric unimodal probability density, 
including standard distributions such as the normal $\mathcal{N}(0,1)$, 
for which $\max_{x}\{|x|\phi(x)\}= e^{-\frac12}/{\sqrt{2\pi}} < 1/2$. Its general validity makes it an essential tool for the proof of our main theorem.

\begin{theorem}\label{sym unimodal}
Consider a Metropolis-Hastings chain with a symmetric random-walk proposal density 
and a symmetric unimodal target density with mean $\mu$. 
Let the centered cumulative distribution function be defined by $\Pi_\mu(x):=\Pi(x+\mu)$. 
Then the unit-lag covariance admits the explicit representation
\begin{align}\label{formula sym unimodal}
\mathrm{Cov}(X_t,X_{t+1})
= 4 \int_0^\infty x \,[1-\Pi_\mu(x)]\,[1-4x\phi(2x)]\,\df x.
\end{align}
\end{theorem}
\begin{proof}
First define
\begin{align*}
f(x):=\int \min\{\pi(x+y),\pi(y)\}\,\df y.
\end{align*}
By Corollary~\ref{invariance}, we may assume without loss of generality that $\mu=0$, then consider \(f(-x)\)
\[
\begin{aligned}
f(-x)
&=\int \min\{\pi(-x+y),\pi(y)\}\,\mathrm{d}y
\\
&\overset{y\mapsto -y}{=}\int \min\{\pi(x+y),\pi(-y)\}\,\mathrm{d}y
= f(x)
\end{aligned}
\]
Hence $f$ is even, i.e., $f(x)=f(-x)$, $\forall x\in\mathbb{R}$.

Recall from equation~\eqref{formula for srw} that
\begin{align}\label{even used}
&\mathrm{Cov}(X_t,X_{t+1})\nonumber\\
&= \sigma_{\pi}^2
-\tfrac{1}{2}\iint (x-y)^2 \phi(y-x)\min\{\pi(x),\pi(y)\}\,\df x\,\df y \nonumber\\
&= \sigma_{\pi}^2
-\tfrac{1}{2}\int x^2 \phi(x)\int \min\{\pi(x+y),\pi(y)\}\,\df y\,\df x\nonumber\\
&=\sigma_{\pi}^2
-\int_0^\infty x^2 \phi(x)\int \min\{\pi(x+y),\pi(y)\}\,\df y\,\df x.
\end{align}
where the second equality follows from a change of variables together with Tonelli's theorem \cite{Tonelli} and the third equality follows from the fact that \(f(\cdot)\) is even. 

Given that the target density is symmetric unimodal, now consider \(f(x),\,x\geq 0\)
\begin{align*}
f(x) 
&= \int_{-\infty}^{-x/2} \pi(y)\,\df y \;+\; \int_{-x/2}^{\infty} \pi(x+y)\,\df y, \qquad x\ge 0 \nonumber\\
&= \Pi(-x/2)+1-\Pi(x/2)\nonumber\\
&= 2[1-\Pi(x/2)].
\end{align*}
Hence,
\begin{align}
\mathrm{Cov}(X_t,X_{t+1})\nonumber
&=\sigma_{\pi}^2\label{w.o. var}
-2\int_0^\infty x^2 \phi(x)[1-\Pi(x/2)]\,\df x\nonumber\\
&=\sigma_{\pi}^2
-16\int_0^\infty x^2 \phi(2x)[1-\Pi(x)]\,\df x.
\end{align}
We can also express the target variance $\sigma_{\pi}^2$ via integration by parts as
\begin{align}\label{variance sym uni}
\sigma_{\pi}^2= 2\int_0^\infty x^2 \pi(x)\,\df x
= 4\int_0^\infty x\,[1-\Pi(x)]\,\df x.
\end{align}
Substituting \eqref{variance sym uni} into \eqref{w.o. var} yields \eqref{formula sym unimodal}.

Now we can readily extend the result to the case of any arbitrary symmetric unimodal density with nonzero mean.

By Corollary~\ref{invariance}, proposal densities of symmetric random–walk form \(q(y|x)=\phi(|y-x|)\) are transition-invariant. Then the covariance admits the same form as in the mean–zero case, with \(\Pi\) replaced by \(\Pi_\mu\).
\end{proof}

Having established the necessary preliminaries, we are now ready to present the main theorem.

\begin{theorem}\label{main theorem}
For a Metropolis-Hastings chain with a symmetric random-walk proposal density 
and a symmetric unimodal target density, the unit-lag covariance is strictly positive.
\end{theorem}
\begin{proof}
Combining Lemma~\ref{leq 0.5} and Theorem~\ref{sym unimodal} yields
\[
\mathrm{Cov}(X_t,X_{t+1})
= 4 \int_0^\infty x \,[1-\Pi_\mu(x)]\,[1-4x\phi(2x)]\,\df x>0
\]
This conclusion holds under our standing assumption that the target distribution is non-degenerate,\,i.e.\,\(\sigma_{\pi}^2>0\).
\end{proof}
We emphasize that MH chains are not necessarily positively correlated. 
If the proposal density deviates from the symmetric random-walk form, 
then even under a symmetric unimodal target distribution, 
the resulting chain can exhibit strictly negative unit-lag correlation, 
as illustrated by the counterexample mentioned in Section~I.
\section{Optimal Proposal Design}
Based on Theorem~\ref{sym unimodal}, we can carry out an optimal proposal design (i.e.\,\,choose the best density \(q(y|x)\)). We define the optimal proposal design as the choice of a symmetric density $\phi$ 
associated with the proposal $q(y|x)=\phi(y-x)$ that minimizes the unit-lag covariance of the Metropolis-Hastings chain under a symmetric unimodal target. This criterion is motivated by the fact that the unit-lag covariance is the leading term in the asymptotic variance expansion of ergodic averages, thereby directly governing estimation efficiency.
Specifically, recall \eqref{formula sym unimodal}
\begin{align}\label{recall prop}
&\mathrm{Cov}(X_t,X_{t+1})\nonumber\nonumber\\
&= 4 \int_0^\infty x \,[1-\Pi_\mu(x)]\,[1-4x\phi(2x)]\,\df x \nonumber\\
&= 4 \int_0^\infty x \,[1-\Pi_\mu(x)]\,\df x
   - 16 \int_0^\infty x^2 \phi(2x)[1-\Pi_\mu(x)]\,\df x.
\end{align}

Only the second term depends on $\phi(\cdot)$, and this dependence is linear. 
Thus, minimizing $\mathrm{Cov}(X_t,X_{t+1})$ is equivalent to maximizing the second term. If we do not impose the standard regularity condition—namely that the support of the target is contained in the support of the proposal, 
$\mathrm{Supp}[\pi(\cdot)] \subseteq \mathrm{Supp}[q(\cdot|x)]$, which ensures the chain is irreducible — 
then the optimal $\phi(\cdot)$ takes the form of a two-point measure, as occurs when the first-order condition admits a unique solution (true for log-concave symmetric unimodal densities)
\[
\phi^*(x) = \tfrac{1}{2}\,\delta(x-x^*) + \tfrac{1}{2}\,\delta(x+x^*),
\]
where \(\delta(\cdot)\) is the Dirac measure at 0, $x^*=2y^*$ and
\[
y^* \in \arg\max_{y\ge 0}\, y^2[1-\Pi_\mu(y)].
\]
Note that $\arg\max$ may define a set, though for symmetric unimodal target densities it commonly reduces to a unique maximizer.
Define \[w(y):=y^2[1-\Pi(y)].\] We can show the following
\begin{align}\label{dont blow up}
    \lim_{y\to\infty}w(y)=\lim_{y\to\infty} y^{2}\bigl(1-\Pi(y)\bigr)=0.
\end{align}
\begin{proof}[Proof of \eqref{dont blow up}]
We have by Markov's inequality that
\[
2x^2[1-\Pi(x)] = x^2\cdot\mathbb{P}(|X| > x)
\le \mathbb{E}[X^{2}; |X| > x].
\]
it suffices to show that $\mathbb{E}[X^{2}; |X| > x] \to 0$ as $x \to \infty$.
Note that $X^{2}\mathbf{1}_{\{|X|>x\}} \to 0$ almost surely and
\(
0 \le X^{2}\mathbf{1}_{\{|X|>x\}} \le X^{2},
\)
with $\mathbb{E}[X^{2}] < \infty$. Hence, by the Dominated Convergence Theorem,
\(
\mathbb{E}[X^{2}; |X| > x] \to 0
\)
\end{proof}

From \eqref{dont blow up}, and noting that $w(0)=0$, $\lim_{y\to\infty}w(y)=0$, and $w(\cdot)$ is continuous, the Extreme Value Theorem guarantees that both the maximum and the minimum (zero) exist. Hence $w_{\max}:=w(y^*)$ must exist for any non-degenerate target distribution with finite variance.

Equivalently, $y^*$ can be determined by the first-order condition
\begin{align}\label{first order condition}
2[1-\Pi_\mu(y)] = y\pi_\mu(y),
\end{align}
with the additional requirement that the second derivative at $y^*$ is negative to ensure local optimality. 

\begin{remark}
By Corollary~\ref{invariance}, we may assume $\mu = 0$ without loss of generality. Equation \eqref{first order condition} has a unique solution when the target density $\pi(\cdot)$ is log-concave.  
In this case, define $g(y):=\pi(y)/(1-\Pi(y))$, which is monotonically increasing, implying that $h(y)=y g(y)$ is strictly increasing. Also note that \(\lim\limits_{y\to+\infty} h(y)
\ge
\lim\limits_{y\to+\infty} y\cdot g(0)
= \infty\)
Hence, the equation $h(y)=2$ admits a solution, and that solution is unique.
\end{remark}

However, under the regularity condition, the optimization problem becomes a supremum problem rather than a maximization problem. In this case, the local maximum achieved by the two-point measure cannot be attained, since attaining it would violate the regularity condition. Proposition~\ref{cant attain} makes this precise by showing that,
when $\phi$ is required to be a probability density satisfying the
support condition needed for irreducibility, the supremum cannot be
attained, although it can be approximated arbitrarily closely.

\begin{proposition}\label{cant attain}
Recall the second term in equation~\eqref{recall prop}. Let $w(y):=y^2[1-\Pi_\mu(y)]$ for $y\ge 0$ and define
\[
\mathcal{J}(\phi):=4\int_{0}^{\infty} w(s)\,\phi(2s)\,\df s,
\]
for any symmetric proposal density $\phi$ about $0$. 
Let $w_{\max}=\sup_{y\ge 0} w(y)$ and 
$\mathcal{M}=\{y\ge 0:\,w(y)=w_{\max}\}$, which is nonempty. 
If $\phi$ is allowed to be a symmetric probability measure, 
the supremum $\sup_\phi\mathcal{J}(\phi)=w_{\max}$ is attained by putting all mass at $\pm 2y$ with $y\in\mathcal{M}$. 
If $\phi$ is required to be a continuous density and $\mathcal{M}$ consists only of isolated points, then the supremum is not attained; 
instead, one can construct a sequence of increasingly concentrated densities around $\{\pm 2y:\,y\in\mathcal{M}\}$ that makes $\mathcal{J}(\phi)$ arbitrarily close to $w_{\max}$.

\end{proposition}

\begin{proof}
By symmetry, with $h(y):=2\phi(2y)$ (noting that $\int_0^\infty h(s)\,\df s = 1/2$), we have
\begin{align*}
\mathcal{J}(\phi)=2\int_{0}^{\infty} w(s)h(s)\,\df s
\le 2\sup_{s\ge0}w(s)\!\int_{0}^{\infty}h(s)\,\df s=w_{\max}
\end{align*}
with equality iff $h(s)=0$ for almost every $s \notin \mathcal{M}$.
Thus any symmetric atomic measure fully supported on $\{\pm 2y:y\in\mathcal{M}\}$ achieves $w_{\max}$.

Now consider the case where $\phi$ is a probability density. Then $h$ is absolutely continuous with respect to the Lebesgue measure.
Since $\mathcal{M}$ consists only of isolated points, it has Lebesgue measure zero.
Consequently, any density $h$ must assign positive mass to the set $\mathcal{M}^c = \{y \ge 0 : y \notin \mathcal{M}\}$.
On this set, $w(y) < w_{\max}$ strictly.
Therefore,
\begin{align*}
    \int_{0}^{\infty} w(s)h(s)\,\df s& = \int_{\mathcal{M}^c} w(s)h(s)\,\df s\\ &< w_{\max} \int_{\mathcal{M}^c} h(s)\,\df s = \frac{w_{\max}}{2}
\end{align*}
which implies $\mathcal{J}(\phi) < w_{\max}$. Thus, the supremum is not attained within the class of densities.

As the concentration increases, the mass shifts towards $\mathcal{M}$, and $\mathcal{J}(\phi)$ approaches $w_{\max}$.
Hence $\sup\mathcal{J}=w_{\max}$.
\end{proof}

Thus, when $\phi(\cdot)$ is required to be a probability density (by the regularity condition), minimizing the unit-lag covariance becomes an infimum problem. Nevertheless, the value $w_{\max}:=\sup_{y\ge0}w(y)$, where $w(y)=y^2[1-\Pi(y)]$ and $\Pi$ is the target cumulative distribution function, can be approximated arbitrarily well. For instance, if $x^*=2y^*$ with $y^*\in\arg\max_{y\ge0}w(y)$, then letting the variance parameter $\sigma^2\to0$ in the Gaussian-component bimodal density
\[
\phi(x) \;=\; \tfrac{1}{2}\,\mathcal{N}(x^*,\sigma^2) \;+\; \tfrac{1}{2}\,\mathcal{N}(-x^*,\sigma^2)\to\phi^*\text{ as }\sigma\to0
\]
yields convergence to the two-point measure $\phi$.

Continuing with the analysis of the unit-lag covariance structure, we next relate it to the ``0.23 rule'' for random-walk Gaussian Metropolis algorithms via the unit-lag covariance matrix. This classic criterion arises from diffusion limit analysis in high dimensions \cite{RobertsandGelman} and prescribes tuning the average acceptance probability to about $0.23$ for optimal efficiency. Here we reinterpret this criterion through the lens of the unit-lag covariance matrix.
We now state Theorem~\ref{0.23}, following the notation and the setup of Roberts et al. \cite{RobertsandGelman}.

\begin{theorem}\label{0.23}
Let $d\in\mathbb{N}$ and consider the product target
\[
\pi_d(\boldsymbol{X})=\prod_{i=1}^d p_i(X_i),\qquad \boldsymbol{X}=(X_1,\dots,X_d)^\top\in\mathbb{R}^d,
\]
where each $p_i$ is a strictly positive $C^2$ density on $\mathbb{R}$ with 
$\sigma_i^2 := \mathrm{Var}_{p_i}(X_i) \in (0, \infty)$, 
$m_i := \mathbb{E}_{p_i}\big[ ( (\log p_i)'(X_i) )^2 \big] \in (0, \infty)$, 
and $\mathbb{E}_{p_i}\big[ | (\log p_i)'(X_i) |^3 \big] < \infty$, 
$\mathbb{E}_{p_i}\big[ | (\log p_i)''(X_i) | \big] < \infty$ for all $i$.
Consider a stationary random-walk Metropolis chain with the symmetric random walk Gaussian proposal density,\,i.e.\,\(\boldsymbol{Y}=\boldsymbol{X}+\boldsymbol{Z}\), where $\boldsymbol{Z}\sim \mathcal N\!\big(\boldsymbol{0},\tfrac{\ell^2}{d}\,\boldsymbol{I}_d\big)$, $\text{for some } 0<\ell<\infty$, 
and $\boldsymbol{I}_d$ denotes the $d\times d$ identity matrix.
Define the averages
\(
\bar\sigma_d^2 := \frac{1}{d}\sum_{i=1}^d \sigma_i^2,\,
\bar m_d := \frac{1}{d}\sum_{i=1}^d m_i.
\) Denote by $\bar m := \lim_{d\to\infty}\bar m_d$, assuming the limit exists in $(0,\infty)$.
Then the unit-lag covariance matrix satisfies
\begin{align}\label{cov matrix formula}
&\mathrm{Cov}(\boldsymbol{X}_t,\boldsymbol{X}_{t+1})\nonumber\\
&= \Big(\mathrm{diag}(\sigma_1^2,\dots,\sigma_d^2)
-\frac{\ell^2}{2d}\cdot2\,\Phi\!\big(-\tfrac{\ell\sqrt{\bar m}}{2}\big)\,\boldsymbol{I}_d\Big)
\;+\; \boldsymbol{R}_d
\end{align}
where $\boldsymbol R_d$ is the remainder matrix whose diagonal entries are $o(d^{-1})$ and off-diagonal entries are $O(d^{-2})$, so that $\|\boldsymbol R_d\|_{\mathrm{op}}=O(d^{-1})$, where $\|\cdot\|_{\mathrm{op}}$ denotes the operator (spectral) norm, chosen since it directly controls eigenvalues and hence positive definiteness. In particular, the diagonal entries
\[
\mathrm{Cov}(\boldsymbol{X}_t,\boldsymbol{X}_{t+1})_{ii}
= \sigma_i^2 \;-\; \frac{\ell^2}{2d}\,2\,\Phi\!\Big(-\tfrac{\ell\sqrt{\bar m}}{2}\Big)\;+o(d^{-1})
\]
are strictly positive for all sufficiently large \(d\) and the matrix $\mathrm{Cov}(\boldsymbol{X}_t,\boldsymbol{X}_{t+1})$ is positive definite.
\end{theorem}

\begin{proof}
\noindent By Corollary 1.2 of Roberts et al. \cite{RobertsandGelman} and applying the change of variable \(\boldsymbol{Y}=\boldsymbol{X}+\boldsymbol{Z}\)
\[\lim_{d\to\infty}\mathbb{E}[\alpha(\boldsymbol{X},\boldsymbol{X}+\boldsymbol{Z})] =2\,\Phi\!\big(-\tfrac{\ell\sqrt{\bar m}}{2}\big)\] and \(\boldsymbol{Z}\text{ is independent of }\boldsymbol{X}\)
which implies, for finite large \(d\)
\begin{align}\label{Ealpha}
\mathbb{E}[\alpha(\boldsymbol{X},\boldsymbol{X}+\boldsymbol{Z})] =2\,\Phi\!\big(-\tfrac{\ell\sqrt{\bar m}}{2}\big)+o(1).
\end{align}
Following the same idea as Lemma~\ref{covariance formula}, we can write the unit-lag covariance matrix as
\begin{align}\label{formula high dim}
\mathrm{Cov}(\boldsymbol{X}_t,\boldsymbol{X}_{t+1})
= \boldsymbol \Sigma_{\pi}-\tfrac12\,\mathbb{E}[\alpha(\boldsymbol{X},\boldsymbol{X}+\boldsymbol{Z})\,\boldsymbol{Z}\boldsymbol{Z}^\top]
\end{align}
where $\boldsymbol\Sigma_{\pi}:=\mathbb{E}_{\pi}\!\big[(\boldsymbol{X}-\boldsymbol\mu)(\boldsymbol{X}-\boldsymbol\mu)^\top\big]$ is the covariance matrix of the stationary distribution with mean $\boldsymbol \mu=\mathbb{E}_{\pi}[\boldsymbol{X}]$.

Then fix \(\boldsymbol{X}\) and set $f(\boldsymbol{X,Z}):=\alpha(\boldsymbol{X},\boldsymbol{X}+\boldsymbol{Z})$. We look at \(\mathbb{E}[\alpha(\boldsymbol{X},\boldsymbol{X}+\boldsymbol{Z}) \boldsymbol{Z}\boldsymbol{Z}^\top]\), for $\boldsymbol{Z}\sim\mathcal N(\boldsymbol0,\ell^2 \boldsymbol{I}_d/d)$,
Applying Stein's Lemma \cite{stein} and the assumption \(\mathbb{E}_{p_i}[|(\log p_i)''(X_i)|] < \infty\) gives
\[
\mathbb{E}[Z_i Z_j f(\boldsymbol{X,Z})| \boldsymbol{X}]
=\frac{\ell^4}{d^2}\,\mathbb{E}[\partial_{ij} f(\boldsymbol{X,Z})| \boldsymbol{X}],
\]
\[
\mathbb{E}[Z_i^2 f(\boldsymbol{X,Z})| \boldsymbol{X}]
=\frac{\ell^2}{d}\,\mathbb{E}[f(\boldsymbol{X,Z})| \boldsymbol{X}]+\frac{\ell^4}{d^2}\,\mathbb{E}[\partial_{ii} f(\boldsymbol{X,Z})|\boldsymbol{X} ].
\]
By the tower property, taking full expectation yields
\begin{align}
&\mathbb{E}[\alpha(\boldsymbol{X},\boldsymbol{X}+\boldsymbol{Z}) Z_i Z_j]=O(d^{-2})\label{avg part1}\\
&\mathbb{E}[\alpha(\boldsymbol{X},\boldsymbol{X}+\boldsymbol{Z}) Z_i^2]=\frac{\ell^2}{d}\,\mathbb{E}[\alpha(\boldsymbol{X},\boldsymbol{X}+\boldsymbol{Z})]+O(d^{-2})\label{avg part2}
\end{align}
Note that $\boldsymbol\Sigma_{\pi}=\mathrm{diag}(\sigma_1^2,\dots,\sigma_d^2)$. Using \eqref{avg part1} and\eqref{avg part2},
\begin{align}\label{r_d}
\mathbb{E}\big[\alpha(\boldsymbol X,\boldsymbol X+\boldsymbol Z)\,\boldsymbol Z\boldsymbol Z^\top\big]
=\frac{\ell^2}{d}\,\mathbb{E}[\alpha(\boldsymbol X,\boldsymbol X+\boldsymbol Z)]\,\boldsymbol I_d -2\boldsymbol{R}_d,
\end{align}
also note that 
\[O(d^{-2})+o(d^{-1})=o(d^{-1})\]
so substituting \eqref{Ealpha} and \eqref{r_d} into \eqref{formula high dim} yields \eqref{cov matrix formula} with a remainder matrix
$\|\boldsymbol R_d\|_{\mathrm{op}}=o(d^{-1})$. Now we can write the diagonal entries as
\begin{align}\label{diagonal entries}
\mathrm{Cov}(\boldsymbol{X}_t,\boldsymbol{X}_{t+1})_{ii}
= \sigma_i^2 \;-\; \frac{\ell^2}{2d}\,2\,\Phi\!\Big(-\tfrac{\ell\sqrt{\bar m}}{2}\Big)\;+o(d^{-1})
\end{align}
which is strictly positive for all large $d$ since $\sigma_i^2>0$.

Note that the leading term \(
\boldsymbol\Sigma_{\pi} - \ell^2/d\cdot\Phi\!\Big(-\ell \sqrt{\bar m}/2\Big) \boldsymbol{I}_d\) in \eqref{cov matrix formula}
is a real diagonal matrix, hence Hermitian and the remainder $\boldsymbol R_d$ arises from 
\[
-\tfrac12\Big(\,\mathbb{E}[\alpha(\boldsymbol X,\boldsymbol X+\boldsymbol Z)\,\boldsymbol Z \boldsymbol Z^\top] 
- \tfrac{\ell^2}{d}\,\mathbb{E}[\alpha(\boldsymbol X,\boldsymbol X+\boldsymbol Z)] \boldsymbol I_d\Big)
\]
it follows that $\boldsymbol R_d$ is real symmetric, hence Hermitian as well. Therefore both the main diagonal term and the perturbation matrix \(\boldsymbol{R}_d\) are Hermitian. Then by Weyl’s inequality \cite{Horn_Johnson_2012}, for Hermitian matrices $\boldsymbol A$ and $\boldsymbol E$,
\[
\lambda_{\min}(\boldsymbol A+\boldsymbol E) \;\ge\; \lambda_{\min}(\boldsymbol A) - \|\boldsymbol E\|_{\mathrm{op}}.
\]
In our case $\boldsymbol A$ is the diagonal term and $\boldsymbol E=\boldsymbol R_d$, so
\[
\lambda_{\min}\!\big(\mathrm{Cov}(\boldsymbol X_t,\boldsymbol X_{t+1})\big)
\;\ge\; \min_i\!\Big\{\sigma_i^2 - \tfrac{\ell^2}{d}\cdot\Phi(-\tfrac12 \ell \sqrt{\bar m})\Big\}
- \|\boldsymbol R_d\|_{\mathrm{op}}
\]
Since $\min\limits_i\sigma_i^2>0$, the subtracted terms are of order $ O(d^{-1})$, hence strictly smaller than $\min\limits_i\sigma_i^2$ for sufficiently large $d$. This shows that the covariance matrix is positive definite in high dimension.
\end{proof}
Theorem~\ref{0.23} shows that in the high-dimensional regime, for a product-form target with non-identically distributed components, the diagonal entries of the unit-lag covariance admit the expansion as equation~\eqref{diagonal entries}. Roberts et al. define the following
\[
h(\ell):=\ell^2\,2\Phi\!\Big(-\tfrac12 \ell\sqrt{\bar m}\Big).
\]
Since $h(\ell)$ coincides with the efficiency criterion in the diffusion limit analysis of Roberts et al., the optimizer $\ell^* := \arg\max_{\ell} h(\ell) = 2.38/\sqrt{\bar m}$ yields the well-known optimal acceptance rate $\mathbb{E}[\alpha(\boldsymbol X,\boldsymbol Y)] \approx 0.23$, see \cite{RobertsandGelman}, while simultaneously minimizing the diagonal entries of the unit-lag covariance matrix by Theorem~\ref{0.23}.

\section{Numerical Study}
\begin{figure}[htbp]
  \centering
  \includegraphics[width=\linewidth]{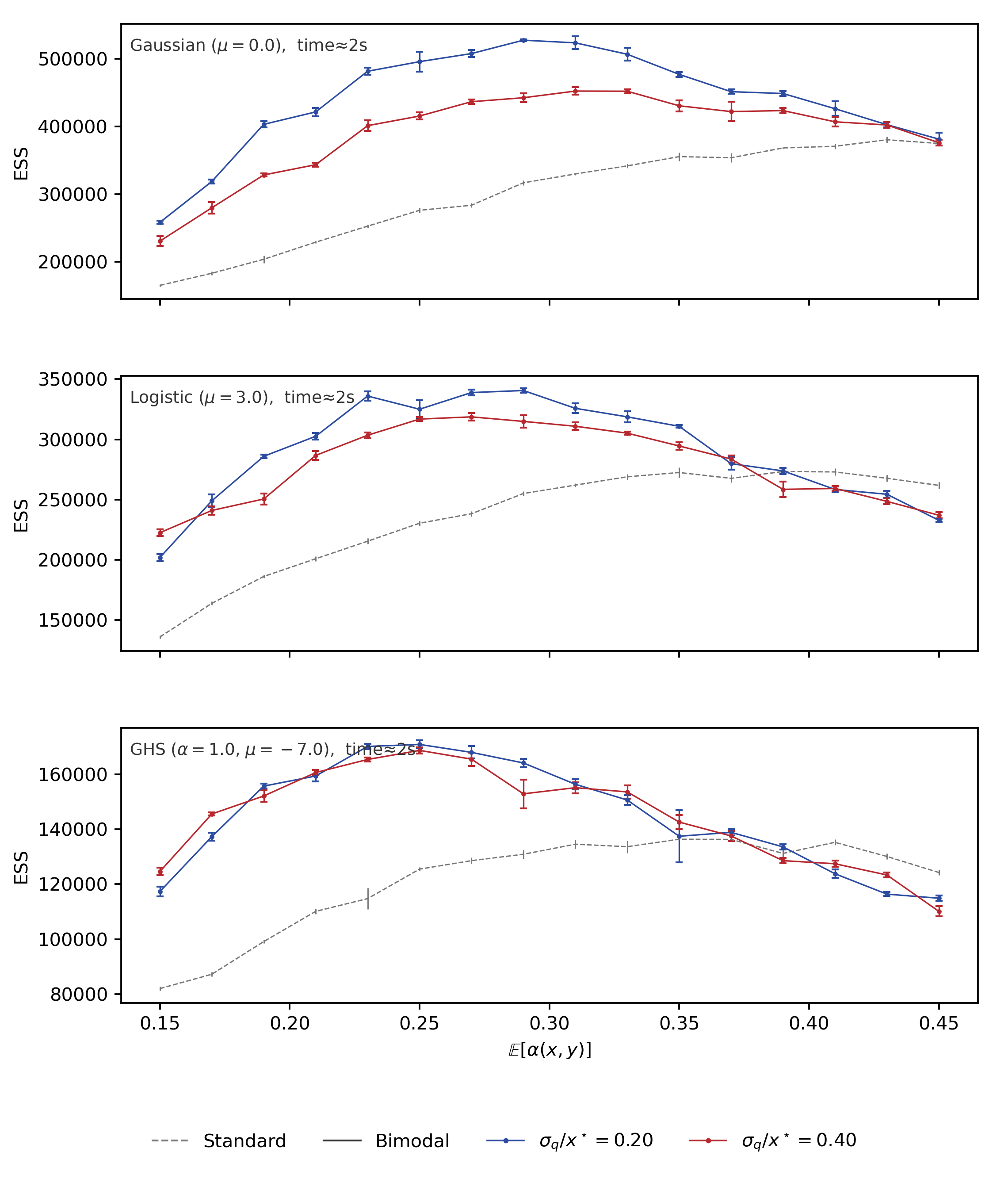}
  \caption{ESS versus $\mathbb{E}[\alpha(x,y)]$ under matched wall-clock time (2s). 
Dashed gray shows standard Gaussian proposals; solid lines show bimodal proposals with $\sigma_q/x^\star=0.20,0.40$.}
  \label{fig:ESS_vs_Ealpha}
\end{figure}
To test the efficiency of symmetric random walk proposal density associated with symmetric bimodal density \(\phi(\cdot)\). We conduct numerical experiments with three symmetric, unimodal target families: (i) the Gaussian, (ii) the logistic, and (iii) the generalized hyperbolic secant (GHS). Their probability density functions are, respectively,
\begin{align}
\pi_{\mathrm{g}}(x| \mu)
&= \frac{1}{\sqrt{2\pi}}\exp\!\left(-\frac{(x-\mu)^2}{2}\right),\\[4pt]
\pi_{\mathrm{ell}}(x| \mu)
&= \frac{e^{-(x-\mu)}}{\bigl(1+e^{-(x-\mu)}\bigr)^2}
= \tfrac{1}{4}\,\operatorname{sech}^2\!\Bigl(\tfrac{x-\mu}{2}\Bigr),\\[4pt]
\pi_{\mathrm{ghs}}(x| \alpha,\mu,\sigma)
&= \frac{c_\alpha}{\sigma}\,
   \operatorname{sech}^{\alpha}\!\Bigl(\frac{\pi\, (x-\mu)}{2\sigma}\Bigr),\,\alpha,\sigma>0,
\end{align}
where \(\operatorname{sech}(z)=1/\cosh(z)\) and the normalizing constant \(c_\alpha\) depends only on \(\alpha\) and admits the closed form
\[
c_\alpha
= \frac{\sqrt{\pi}}{2}\,
  \frac{\Gamma\!\bigl(\tfrac{\alpha+1}{2}\bigr)}{\Gamma\!\bigl(\tfrac{\alpha}{2}\bigr)}
= \left(\frac{2}{\pi}\int_{-\infty}^{\infty}\operatorname{sech}^{\alpha}(t)\,\mathrm{d}t\right)^{-1}.
\]
All three families are symmetric about $\mu$ and unimodal. For the Gaussian and logistic targets with unit scale (without loss of generality), the cumulative distribution functions are $\Phi(x-\mu)$ and $(1+e^{-(x-\mu)})^{-1}$, respectively, while the GHS cdf has no closed form and is evaluated numerically. Unless otherwise specified, we set the Gaussian variance to $1$, the logistic scale to $1$, and $\sigma=1,\,\alpha=1$ for the GHS family.
\begin{figure}[htbp]
  \centering
  \includegraphics[width=\linewidth]{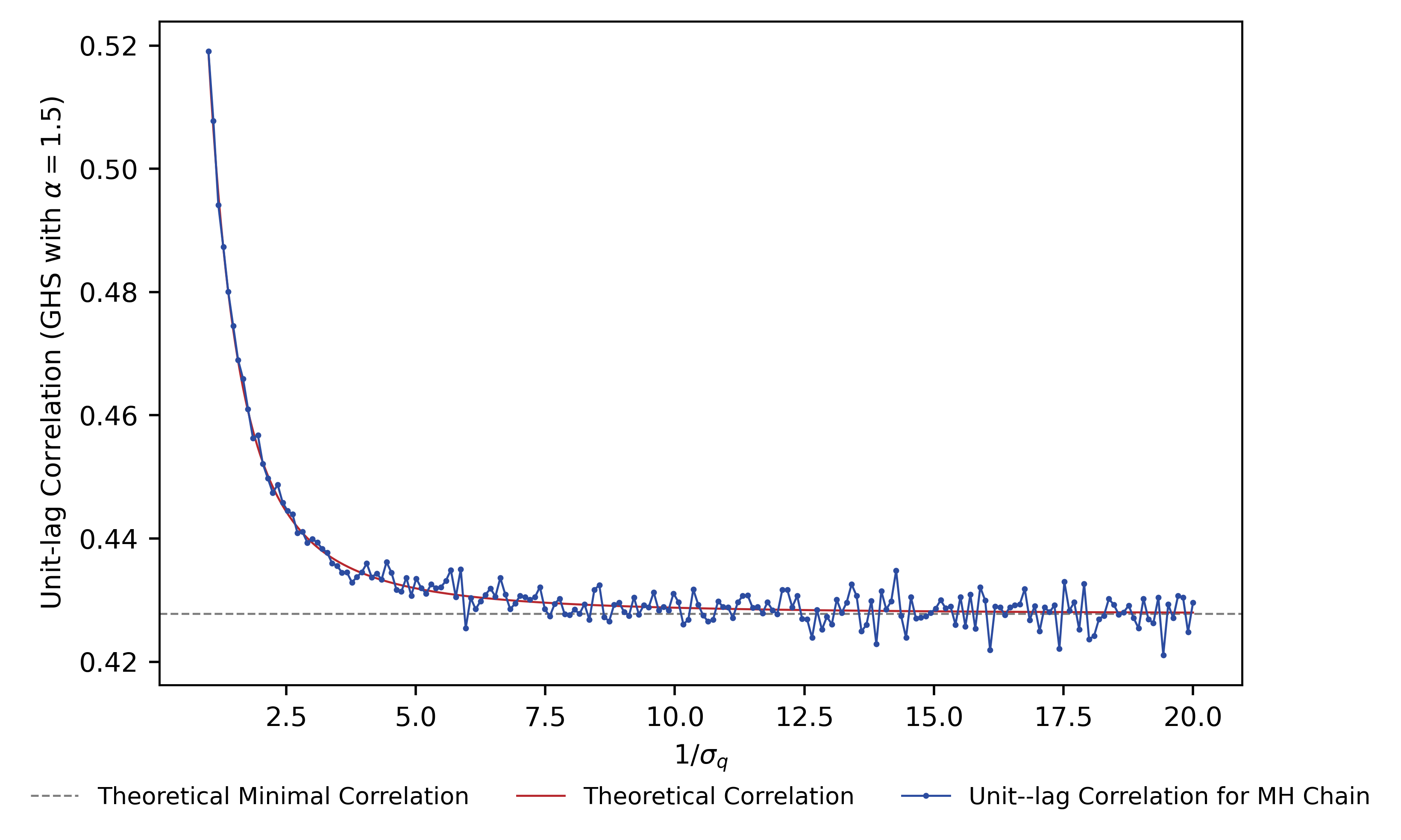}
  \caption{Unit-lag correlation decreases with narrower bimodal proposals, approaching the two-point limit—the unattainable minimum under any continuous proposal density.}
  \label{fig:covariance_vs_sigma}
\end{figure}
For the proposal distribution, we consider two cases: a Gaussian–component symmetric bimodal density 
\(\phi(x)=\tfrac12\mathcal{N}(x^*,\sigma_q^2)+\tfrac12\mathcal{N}(-x^*,\sigma_q^2)\) with \(\sigma_q>0\), 
and a standard Gaussian symmetric unimodal density. To ensure a fair comparison, we evaluate each MH chain in terms of effective sample size (ESS) and unit-lag autocorrelation, matched at the same (or nearly the same) average acceptance rate and running time (In practice, we fixed the wall-clock budget). The ESS is defined as 
\[
\mathrm{ESS}:=\frac{N}{1+2\sum_{k=1}^\infty \rho_k},
\]
where $N$ is the total number of samples and $\rho_k$ is the lag-$k$ autocorrelation \cite{DefofESS}. ESS measures the number of independent samples that the correlated MCMC output is equivalent to, thereby quantifying the impact of autocorrelation on estimation efficiency. Across target densities, the bimodal design typically yields higher ESS and lower unit-lag autocorrelation than the standard Gaussian proposal.

To ensure a fair comparison across different average acceptance rates, we applied a global scaling factor to the proposal distributions. Specifically, for the bimodal proposal, the ratio $\sigma_q/x^*$ was held fixed to preserve the bimodal geometry, while a scalar tuning parameter was adjusted numerically to achieve the target average acceptance rate. This ensures that the performance differences observed in Figure \ref{fig:ESS_vs_Ealpha} are attributable to the distributional shape rather than suboptimal variance scaling.


In Figure~1, we report the effective sample size for three target densities (Gaussian, Logistic, and GHS) under matched average acceptance probability and a runtime of 2 seconds. For the bimodal proposal, two cases with $\sigma_q/x^* = 0.20$ and $\sigma_q/x^* = 0.40$ were tested. In both cases, the narrow bimodal Gaussian yields substantially higher ESS and thus greater sampling efficiency than the standard Gaussian, consistent with concentrating proposal mass near the optimal jumps at $\pm 2x^*$, where $x^*$ maximizes $w(x)$. This advantage, however, diminishes as $\mathrm{E}[\alpha(x,y)]$ approaches $50\%$, since in this regime the chain already accepts nearly half of the proposed moves, and the marginal benefit of concentrating proposal mass near the optimal jumps becomes negligible.



\begin{figure}[htbp]
  \centering
  \includegraphics[width=\linewidth]{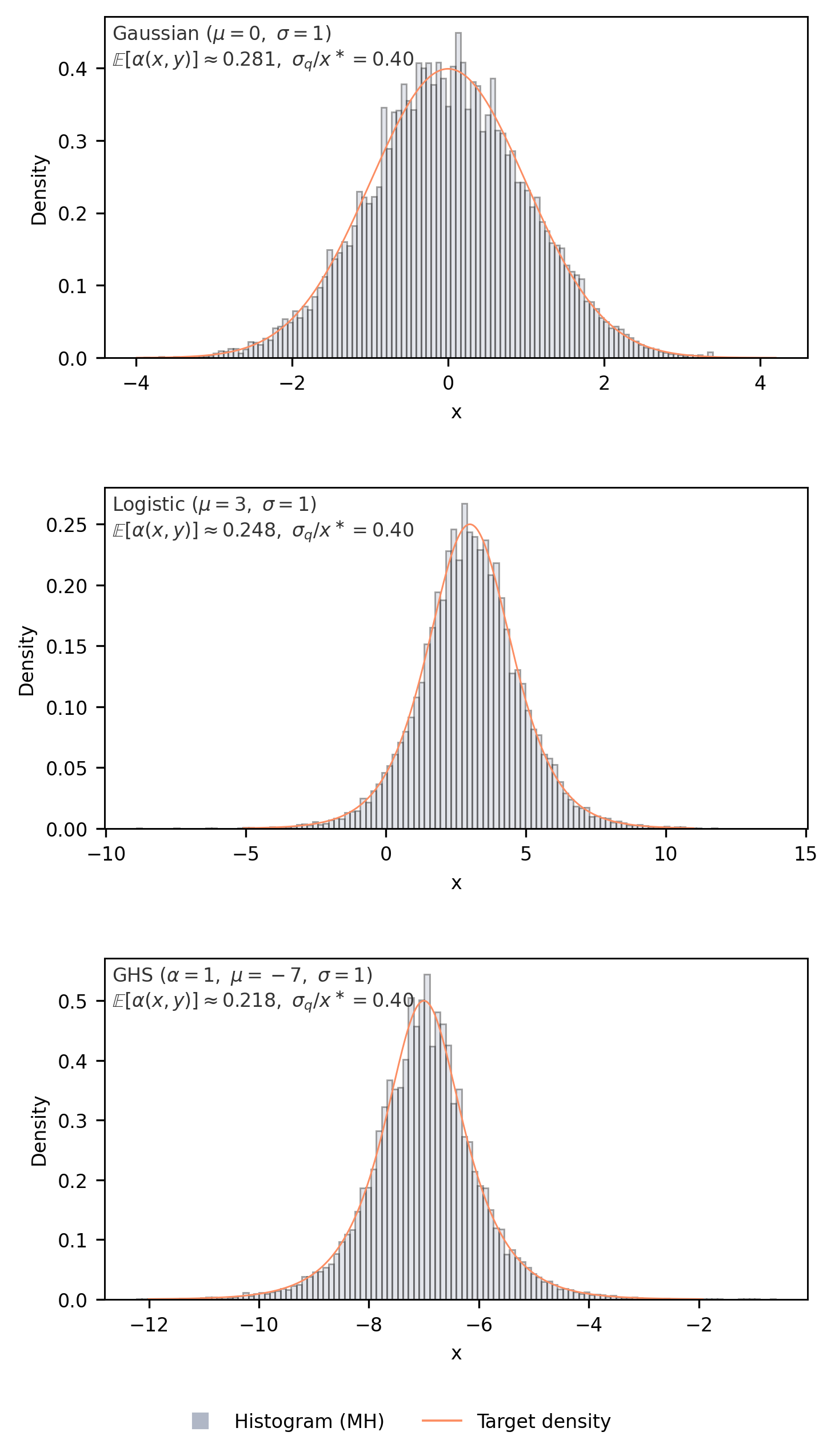}
  \caption{Histograms of MH samples versus target densities for Gaussian $(\mu=0,\sigma=1)$, Logistic $(\mu=3,\sigma=1)$, and GHS $(\alpha=1,\mu=-7,\sigma=1)$. All chains ran for $\approx1$s with bimodal proposal $(\sigma_q/x^\ast=0.40)$.}
  \label{fig:hist_vs_targets}
\end{figure}

In Figure~2, consider the GHS target density with $\mu=0$, variance $1$, and shape parameter $\alpha=1.5$, the theoretical unit-lag correlation (equal here to the unit-lag covariance) converges, as $1/\sigma_q$ increases,\,i.e.\,\,\(\sigma_q\) decreases, to the minimum value obtained by substituting $\phi^*$ into equation~\eqref{formula sym unimodal}. The empirical unit-lag autocorrelation of the MH chain closely tracks the theoretical curve and fluctuates around this minimum, as seen in Figure~2. These oscillations reflect numerical variation, which is more pronounced for small $\sigma_q$ due to higher sample correlation and lower effective sample size.


In Figure~3, we present histograms of MH samples against the true target densities for three cases (Gaussian, Logistic, and GHS) under the $\sigma_q/x^*=0.40$. The close alignment between the empirical histograms and the target curves confirms that the MH algorithm using the narrow bimodal density \(\phi\) accurately recovers the stationary distribution,\,i.e.\,the target distribution.

\section{Concluding Remarks and Future Work}
Understanding the autocovariance structure of the Metropolis–-Hastings algorithm is essential for assessing its efficiency. In this paper, we presented new theoretical results for random-walk MH chains targeting symmetric unimodal densities and introduced an optimal proposal design based on the unit-lag covariance, ultimately relating the analysis to high-dimensional settings in an asymmetric sense. Future work will focus on relaxing the restrictive assumptions on target distributions and extending the framework to multi-lag covariance structures.

\section*{Acknowledgment}
The author thanks Dr.\,James\,A.\,Fill for an earlier proof of Lemma~\ref{leq 0.5} and for helping us to formulate the statements and proofs of the results in Section 3 (especially Theorems 3.4--3.5).
\balance
\bibliographystyle{IEEEtran}
\bibliography{root}

@ARTICLE{HillandSpall,
  author={Hill, Stacy D. and Spall, James C.},
  journal={IEEE Control Systems Magazine}, 
  title={Stationarity and Convergence of the {M}etropolis-{H}astings Algorithm: Insights into Theoretical Aspects}, 
  year={2019},
  volume={39},
  number={1},
  pages={56-67},
  doi={10.1109/MCS.2018.2876959}}

@article{RobertsandGelman,
 ISSN = {10505164},
 author = {G. O. Roberts and A. Gelman and W. R. Gilks},
 journal = {The Annals of Applied Probability},
 number = {1},
 pages = {110--120},
 publisher = {Institute of Mathematical Statistics},
 title = {Weak Convergence and Optimal Scaling of Random Walk {M}etropolis Algorithms},
 volume = {7},
 year = {1997}
}

@article{Tierney1994,
 ISSN = {00905364, 21688966},
 author = {Luke Tierney},
 journal = {The Annals of Statistics},
 number = {4},
 pages = {1701--1728},
 publisher = {Institute of Mathematical Statistics},
 title = {Markov Chains for Exploring Posterior Distributions},
 volume = {22},
 year = {1994}
}

@article{RobertsandRosenthal,
author = {Gareth O. Roberts and Jeffrey S. Rosenthal},
title = {{General state space Markov chains and MCMC algorithms}},
volume = {1},
journal = {Probability Surveys},
number = {none},
publisher = {Institute of Mathematical Statistics and Bernoulli Society},
pages = {20 -- 71},
year = {2004},
doi = {10.1214/154957804100000024},
}

@article{BrooksandGareth1998,
  author    = {Stephen P. Brooks and Gareth O. Roberts},
  title     = {Convergence assessment techniques for {M}arkov chain {M}onte {C}arlo},
  journal   = {Statistics and Computing},
  volume    = {8},
  number    = {4},
  pages     = {319--335},
  year      = {1998},
  doi       = {10.1023/A:1008820505350},
}

@ARTICLE{Spall2003,
  author={Spall, James C.},
  journal={IEEE Control Systems Magazine}, 
  title={Estimation via {M}arkov chain {M}onte {C}arlo}, 
  year={2003},
  volume={23},
  number={2},
  pages={34-45},
  doi={10.1109/MCS.2003.1188770}}

@article{DefofESS,
 ISSN = {00031305, 15372731},
 author = {Robert E. Kass and Bradley P. Carlin and Andrew Gelman and Radford M. Neal},
 journal = {The American Statistician},
 number = {2},
 pages = {93--100},
 publisher = {[American Statistical Association, Taylor & Francis, Ltd.]},
 title = {Markov Chain {M}onte {C}arlo in Practice: A Roundtable Discussion},
 volume = {52},
 year = {1998}
}

@article{Godsill2001,
  author    = {Godsill, Simon and Doucet, Arnaud and West, Mike},
  title     = {Maximum a Posteriori Sequence Estimation Using {M}onte {C}arlo Particle Filters},
  journal   = {Annals of the Institute of Statistical Mathematics},
  year      = {2001},
  volume    = {53},
  number    = {1},
  pages     = {82--96},
  doi       = {10.1023/A:1017968404964},
  issn      = {1572-9052}
}

@article{Golightly2005,
 ISSN = {0006341X, 15410420},
 author = {A. Golightly and D. J. Wilkinson},
 journal = {Biometrics},
 number = {3},
 pages = {781--788},
 publisher = {[Wiley, International Biometric Society]},
 title = {Bayesian Inference for Stochastic Kinetic Models Using a Diffusion Approximation},
 volume = {61},
 year = {2005}
}

@article{Frigola2013,
author = {Frigola, Roger and Lindsten, Fredrik and Schön, Thomas and Rasmussen, Carl},
year = {2013},
month = {06},
pages = {},
title = {Bayesian Inference and Learning in Gaussian Process State-Space Models with Particle {MCMC}},
journal = {Advances in Neural Information Processing Systems}
}

@article{Ninness2010,
title = {Bayesian system identification via {M}arkov chain {M}onte {C}arlo techniques},
journal = {Automatica},
volume = {46},
number = {1},
pages = {40-51},
year = {2010},
issn = {0005-1098},
doi = {https://doi.org/10.1016/j.automatica.2009.10.015},
author = {Brett Ninness and Soren Henriksen}
}

@article{Tonelli,
  author    = {A. Mukherjea},
  title     = {A remark on {T}onelli's theorem on integration in product spaces},
  journal   = {Pacific Journal of Mathematics},
  volume    = {42},
  number    = {1},
  pages     = {177--185},
  year      = {1972},
  publisher = {Pacific Journal of Mathematics},
}

@INPROCEEDINGS{Srinivasan,
  author={Srinivasan, Anand and Takeishi, Naoya},
  booktitle={2020 59th IEEE Conference on Decision and Control (CDC)}, 
  title={An {MCMC} Method for Uncertainty Set Generation via Operator-Theoretic Metrics}, 
  year={2020},
  volume={},
  number={},
  pages={2714-2719},
  doi={10.1109/CDC42340.2020.9304222}}

@INPROCEEDINGS{Moriarty,
  author={Moriarty, John and Vogrinc, Jure and Zocca, Alessandro},
  booktitle={2018 IEEE Conference on Decision and Control (CDC)}, 
  title={Frequency violations from random disturbances: an {MCMC} approach}, 
  year={2018},
  volume={},
  number={},
  pages={1598-1603},
  doi={10.1109/CDC.2018.8619304}}

@INPROCEEDINGS{Zhao,
  author={Zhao, Zinan and Kumar, Mrinal},
  booktitle={2014 American Control Conference}, 
  title={A {MCMC}/{B}ernstein approach to chance constrained programs}, 
  year={2014},
  volume={},
  number={},
  pages={4318-4323},
  doi={10.1109/ACC.2014.6859159}}

@inbook{Horn_Johnson_2012,
  author    = {Horn, Roger A. and Johnson, Charles R.},
  title     = {Matrix Analysis},
  edition   = {2},
  publisher = {Cambridge University Press},
  year      = {2012},
  pages     = {258--261},
  address   = {Cambridge}
}

@article{Metropolis1953,
author = {Siddhartha Chib and Edward Greenberg},
title = {Understanding the {M}etropolis-{H}astings Algorithm},
journal = {The American Statistician},
volume = {49},
number = {4},
pages = {327--335},
year = {1995},
publisher = {ASA Website},
doi = {10.1080/00031305.1995.10476177},
eprint = {       https://www.tandfonline.com/doi/pdf/10.1080/00031305.1995.10476177
}
}

@article{stein,
author = {Lin, Wu and Khan, Mohammad and Schmidt, Mark},
year = {2019},
month = {10},
pages = {},
title = {Stein's Lemma for the Reparameterization Trick with Exponential Family Mixtures},
doi = {10.48550/arXiv.1910.13398}
}
\end{document}